\documentclass[10pt,romanappendices,conference,twocolumn]{IEEEtran}
\usepackage{times}
\usepackage{latexsym}
\usepackage{stfloats}
\usepackage{citesort}
\usepackage{graphicx,amsfonts,amsmath,amssymb,array,url}
\usepackage[noadjust]{cite}
\usepackage{amsthm}
\usepackage{extarrows}
\usepackage{ifpdf}
\usepackage{cite}
\usepackage{epstopdf}
\usepackage{algorithmic}
\usepackage[hidelinks]{hyperref} 
\usepackage{epsfig}

\DeclareMathOperator{\sinc}{sinc}
\theoremstyle{remark} 

\newtheorem{Lindeberg}{Lemma}
\newtheorem{F1}[Lindeberg]{Lemma}
\newtheorem{Fn}[Lindeberg]{Lemma}
\newtheorem{FH}[Lindeberg]{Lemma}
\newtheorem{FHextension}[Lindeberg]{Lemma}
\begin{document}
\title{Massive MIMO Relaying with Hybrid Processing}

\author{\IEEEauthorblockN{Milad Fozooni\IEEEauthorrefmark{1}, Michail Matthaiou\IEEEauthorrefmark{1}, Shi Jin\IEEEauthorrefmark{2}, and George C. Alexandropoulos\IEEEauthorrefmark{3}}
\IEEEauthorblockA{\IEEEauthorrefmark{1}School of Electronics, Electrical Engineering and Computer Science, Queen's University Belfast, Belfast, U.K.}
\IEEEauthorblockA{\IEEEauthorrefmark{2}National Mobile Communications Research Laboratory, Southeast University, Nanjing, China}
\IEEEauthorblockA{\IEEEauthorrefmark{3}Mathematical and Algorithmic Sciences Lab, France Research Center, Huawei Technologies Co. Ltd., Paris, France}
Emails: \{mfozooni01, m.matthaiou\}@qub.ac.uk,  jinshi@seu.edu.cn, george.alexandropoulos@huawei.com }
\maketitle
\begin{abstract}
Massive multiple-input multiple-output (MIMO) \mbox{relaying} is a promising technological paradigm which can offer high spectral efficiency and substantially improved coverage. Yet, these configurations face some formidable challenges in terms of digital signal processing (DSP) power consumption and circuitry complexity, since the number of radio frequency (RF) chains may scale with the number of antennas at the relay station. In this paper, we advocate that performing a portion of the power-intensive DSP in the analog domain, using simple phase shifters and with a reduced number of RF paths, can address these challenges. In particular, we consider a multipair amplify-and-forward (AF) relay system with maximum ratio \mbox{combining}/transmission (MRC/MRT) and we determine the asymptotic spectral efficiency for this hybrid analog/digital architecture. After that, we extend our analytical results to account for heavily quantized analog phase shifters and show that the performance loss with 2 quantization bits is only $10 \%$.
\end{abstract}

\section{Introduction}
\label{Introduction}
Massive MIMO is a promising way to reap all advantages of a MIMO system, such as power and multiplexing gains in a larger scale \cite{ngo2014multipair, rusek2013scaling, marzetta2010noncooperative}. It has also been extensively investigated over the past years, thanks to its ability to cancel out noise, inter-user interference and fast fading. Fortunately, all these advantages can be obtained with simple linear signal processing \cite{marzetta2010noncooperative,suraweera2013multi,zhang2009designing}. On the other hand, MIMO relay systems have been intensively studied since they can provide extended coverage and enhance the spectral efficiency, particularly at the edges of cells \cite{dohler2010cooperative}. However, they typically require an extremely complex power allocation and precoder/decoder design \cite{jin2015ergodic,rong2009unified,phan2009power,cao2012regenerative}. Therefore, a relaying system with a massive number of antennas at the relay station has emerged as a viable candidate to address the aforementioned challenges.

Massive relaying is a fairly new research area which has been investigated from different viewpoints. In \cite{ngo2014multipair}, a massive relay is considered to overcome the detrimental effects of loop interference in full-duplex operation. There are also some other research efforts which investigate the spectral efficiency of massive relaying and derive asymptotic scaling laws \cite{suraweera2013multi,jin2015ergodic,ngo2013spectral}. However, having one RF chain dedicated to each antenna imposes several challenges in terms of DSP power consumption and circuitry complexity such that this fully digital architecture may scale badly especially in the mm-wave regime \cite{el2014spatially}. Recently, this critical issue has been addressed by researchers in other fields \cite{el2014spatially, ni2015near, liang2014low, ying2015hybrid} and many scholars hold the view that the best suitable solution is a hybrid structure consisting of a digital baseband processor and an analog RF beamformer/combiner. A considerable amount of literature assumes a hybrid analog/digital transceivers for different communications applications \cite{el2014spatially, ni2015near, liang2014low, ying2015hybrid,alkhateeb2013hybrid, molisch2004fft}, but not in the context of relaying. More recently, \cite{lee2014af} assumed a half-duplex relay system where each node is equipped with a hybrid beamformer, but hybrid processing is performed on the nodes not the relay, while no spectral efficiency characterization is being presented either. Motivated by the above discussion, this paper investigates, for the first time ever, the performance of a multipair massive relaying where part of the DSP on the relay station is performed in the analog domain, using simple analog phase shifters. In particular, we analytically determine the asymptotic end-to-end spectral efficiency by considering MRC/MRT processing, where the number of antennas grows up without bound. Then, we elaborate on three power saving strategies and deduce their asymptotic power scaling laws. These laws reveal important physical insights and tradeoffs between the transmit power of user nodes and relay. Finally, we consider the case of quantized phase shifters and work out the performance degradation for small number of quantization bits. Our numerical results indicate that (a) hybrid processing can offer a very satisfactory performance with a substantially lower power consumption and number of RF chains; and (b) 2 bits of quantization cause a minor performance degradation (approximately $10 \%$).
\begin{figure*}[!t]
	\centering
		\includegraphics[width=6.8 in]{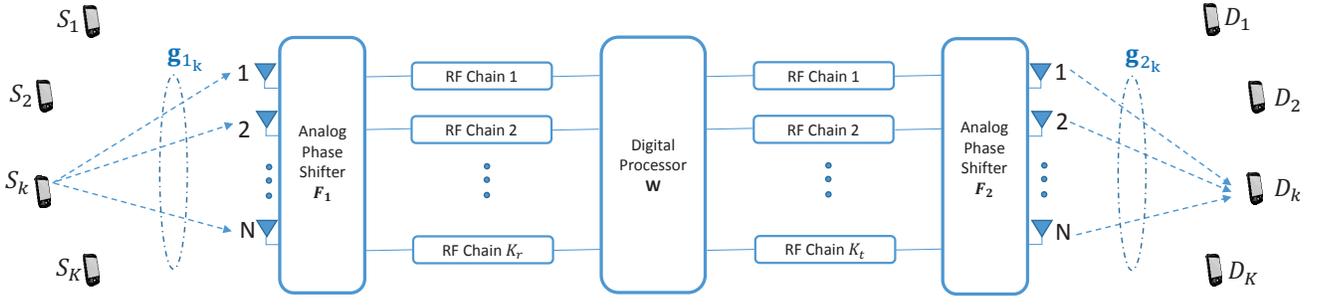}
		\caption{Block diagram of a massive relay system with a baseband digital processor 
		combined with two analog RF beamformers which are implemented using quantized phase shifters.}
\label{SystemModel}
\end{figure*}

\par \textit{Notation}: Upper and lower case bold-face letters denote matrices and vectors, respectively. Also, the symbols $(\cdot)^T$,  $(\cdot)^*$, $(\cdot)^H$, $\mathrm{Tr}(\cdot)$, $\left\| \cdot \right\|$, and $\left\| \cdot \right\|_{F}$ indicate the transpose, conjugate, conjugate transpose, trace operator, Euclidean norm, and Frobenius norm, respectively. In addition, the symbol $\left[\cdot\right]_{m,n}$ returns the \mbox{$(m,n)$-th} element of a matrix. Also, we define the phase and absolute value of a complex number $z$ with $\angle z$ and $|z|$, respectively. Furthermore, $\mathbb {E}\left[\cdot\right]$ is the expectation operation, and $\mathbf{I}_N$ refers to the $N\times N$ identity matrix.  
\section{System Model}
Consider a system model as shown in Fig. \ref{SystemModel}, where a group of $K$ sources, $S_k$ with $k=1,2,\ldots,K$, communicate with their own destinations, $D_k$, via a single one-way relay, $R$. All sources and destinations are equipped with a single antenna while the relay is equipped with $N$ antennas on each side. Furthermore, the direct link among the $K$ pairs does not exist due to large path loss and heavy shadowing; to keep our analysis simple, full channel state information (CSI) is available and we ignore hardware imperfections \cite{Larsson2014MassivesTutorial}. Users send their data streams through a narrowband flat-fading propagation channel in the same time-frequency block. To keep the implementation cost of this massive MIMO relaying topology at low levels, we consider $K_r$ receive and $K_t$ transmit RF chains at the relay, with $K_r, K_t \ll N$. As mentioned above, by reducing the number of RF paths, we can avail of reduced power consumption (reduced numbers of power amplifiers and analog-to-digital converters) and reduced circuitry. Moreover, to reduce the power dissipation of DSP, we deploy an analog combiner $\mathbf{F}_1 \in \mathbb{C}^{K_r \times N}$ and precoder $\mathbf{F}_2 \in \mathbb{C}^{K_t \times N}$ at the relay station which perform phase matching at a much lower dimension compared to full DSP. Since analog processing alone is not flexible enough, the remaining fraction of signal processing is performed in the digital domain through the smaller dimensional matrix $\mathbf{W}\in \mathbb{C}^{K_t \times K_r}$. Under this model, the received signal at the relay and destination can be mathematically expressed, respectively, as 
\begin{align}
\mathbf{y}_R&=\sqrt{P_u}{\mathbf{G}}_1\mathbf{x}+\mathbf{n}_R
\\
\mathbf{y}_D&= \sqrt{P_u}\mathbf{G}_{2}^H\mathbf{F}_{2}^H \mathbf{W}\mathbf{F}_1\mathbf{G}_1\mathbf{x}+\mathbf{G}_2^H\mathbf{F}_2^H\mathbf{W}\mathbf{F}_1\mathbf{n}_R+\mathbf{n}_D
\label{Matrix Model}
\end{align}
where $P_u$ represents the transmitted power of each source, and $\mathbf{x}=[x_1, x_2, \cdots, x_K]^T$ is the zero-mean Gaussian symbol vector such that $\mathbb{E}\big[\mathbf{x}\mathbf{x}^H\big]=\mathbf{I}_K$. 
Also, the received signal at the destinations is included in $\mathbf{y}_D \in \mathbb{C}^{K \times 1}$, while the $N$-dimensional vector $\mathbf{n}_R$ and \mbox{$K$-dimensional} vector $\mathbf{n}_D$ model the additive circularly symmetric complex Gaussian noise such that \mbox{$\mathbf{n}_R \sim \mathcal{CN}(0,\sigma_{n_R}^2 \mathbf{I}_{N})$} and $\mathbf{n}_D \sim \mathcal{CN}(0,\sigma_{n_D}^2 \mathbf{I}_{K})$. Moreover, ${\mathbf{G}}_1 \in \mathbb{C}^{N \times K}$ and ${\mathbf{G}}_2 \in \mathbb{C}^{N \times K}$ express the propagation channel between sources and relay, and between relay and destinations, respectively. More precisely, \mbox{${\mathbf{G}}_1={\mathbf{H}}_1{\mathbf{D}}_1^{\frac{1}{2}}$} and ${\mathbf{G}}_2={\mathbf{H}}_2{\mathbf{D}}_2^{\frac{1}{2}}$, where ${\mathbf{H}}_1$, ${\mathbf{H}}_2 \in \mathbb{C}^{N \times K}$ refer to small-scale fading channels with independent and identically distributed (i.i.d.) entries, each of them following $\mathcal{CN}(0,1)$. The diagonal matrices $\mathbf{D}_1$ and $\mathbf{D}_2 \in \mathbb{C}^{K \times K}$ include the \mbox{large-scale fading parameters}, where we define $\eta_{1,k}\stackrel{\Delta}{=}\big[{\mathbf{D}}_1\big]_{k,k}$ and $\eta_{2,k}\stackrel{\Delta}{=}\big[{\mathbf{D}}_2\big]_{k,k}$. From (\ref{Matrix Model}) the received signal at the $k$-th destination is given by
\begin{align}
\label{received signal}
y_{D_k}=&\sqrt{P_u}\mathbf{g}_{2_k}^H\mathbf{F}_2^H\mathbf{W}\mathbf{F}_1\mathbf{g}_{1_k}x_k+\sqrt{P_u}\sum_{i\neq k}^K \mathbf{g}_{2_k}^H\mathbf{F}_2^H\mathbf{W}\mathbf{F}_1\mathbf{g}_{1_i}x_i \nonumber
\\+&\mathbf{g}_{2_k}^H\mathbf{F}_2^H\mathbf{W}\mathbf{F}_1\mathbf{n}_R+n_{D_k}
\end{align}
where $\mathbf{g}_{1_k}$, and $\mathbf{g}_{2_k}$ denote the $k$-th column of the matrices $\mathbf{G}_1$ and $\mathbf{G}_2$, respectively. In (\ref{received signal}), the first term corresponds to the desired signal, the second term refers to the interpair-interference, while the last two terms correspond to the amplified noise at the relay and noise at the destination, respectively. Thus, the instantaneous end-to-end signal-to-interference-noise ratio (SINR) for the $k$-th pair is given by
\begin{align}
\label{Monte Carlo}
{\rm SINR}_k\hspace{-2pt}=\hspace{-2pt}\frac{{P_u}\Big|\mathbf{g}_{2_k}^H\mathbf{F}_2^H\mathbf{W}\mathbf{F}_1\mathbf{g}_{1_k}\Big|^2}
{\hspace{-2pt}P_u\hspace{-2pt}\sum\limits_{i\neq k}^K \Big|\mathbf{g}_{2_k}^H\mathbf{F}_2^H\mathbf{W}\mathbf{F}_1\mathbf{g}_{1_i}\Big|^2
\hspace{-4pt}+\hspace{-2pt}\|\mathbf{g}_{2_k}^H\mathbf{F}_2^H\mathbf{W}\mathbf{F}_1\|^2\sigma_{n_R}^2
\hspace{-2pt}+\hspace{-2pt}\sigma_{n_D}^2}.
\end{align}
Consequently, the average spectral efficiency (bits/s/Hz) of this multipair massive MIMO relaying system can be obtained as
\begin{align}
R=\frac{1}{2}\sum\limits_{k=1}^{K}{\mathbb{E}\Big[\log_2\left({1+{\rm SINR_k}}\right)}\Big]
\end{align} 
where the pre-log factor $1/2$ is due to the half-duplex relaying. As mentioned before, the role of the analog combiners is to balance out the phase of the propagation matrices. It is noteworthy that the matrices $\mathbf{F}_1$ and $\mathbf{F}_2$ can only perform analog phase shifting, hence, their elements amplitude are assumed to be fixed by $1/\sqrt{N}$. To this end, we have
\begin{align}
\angle \big[\mathbf{F}_1\big]_{i,j}&=-\angle\big[\mathbf{G}_1\big]_{j,i} \nonumber \\
\angle \big[\mathbf{F}_2\big]_{i,j}&=-\angle \big[\mathbf{G}_2\big]_{j,i} \\
\Big|\big[\mathbf{F}_1\big]_{i,j}\Big|&= \Big|\big[\mathbf{F}_2\big]_{i,j}\Big|=\frac{1}{\sqrt{N}}. \nonumber 
\end{align}
On the other hand, the baseband precoder matrix $\mathbf{W}$ can modify both the amplitude and phase of the incoming vector. Moreover, we introduce the following long-term transmit power constraint for the output of the relay station: 
\begin{align}
\label{constraint}
{\rm Tr} \left({\mathbb{E}\left[\tilde{\mathbf{y}}_R\tilde{\mathbf{y}}_R^H\right]}\right)= P_r 
\end{align}
where, $\tilde{\mathbf{y}}_R=\mathbf{F}_2^H\mathbf{W}\mathbf{F}_1\mathbf{y}_R$ demonstrates the combined relay output signal. 
In the rest of this paper, we assume MRC to combine received signals at the relay, and also consider MRT to forward the received signals from the relay to the destinations. We recall that MRC/MRT type of processing has been well integrated in the context of massive MIMO, since it offers a near-optimal performance and can be implemented in a distributed manner \cite{ngo2013energy}.

We now define the following symbols that will be used in our subsequent analysis; $\mathbf{A}_1\stackrel{\Delta}{=}\mathbf{F}_1\mathbf{G}_1$, and
$\mathbf{A}_2\stackrel{\Delta}{=}\mathbf{F}_2\mathbf{G}_2$.
Regarding the digital MRC/MRT transformation \mbox{matrix}, \mbox{$\mathbf{W}\stackrel{\Delta}{=}\alpha \mathbf{A}_2\mathbf{A}_1^H$} where $\mathbf{\alpha}$ is a normalization constant that guarantees that the power constraint in (\ref{constraint}) is satisfied. Therefore, we can obtain
after some mathematical manipulations
\begin{align}
\label{alpha}
\alpha=\sqrt{\frac{P_r}{P_u\|\mathbf{F}_2^H\mathbf{A}_2 \mathbf{A}_1^H\mathbf{A}_1\|_{F}^2  +\sigma_{n_R}^2 \|{\mathbf{F}_2^H\mathbf{A}_2\mathbf{A}_1^H\mathbf{F}_1\|}_F^2}}.
\end{align} 

\vspace{2pt}
\section{Large $N$ analysis}
In this section, we asymptotically analyze the performance of the massive MIMO relay with hybrid processing in two dedicated subsections: section \ref{ideal phase shifters} assuming ideal (continuous) phase shifters, and section \ref{quantized phase shifters} assuming phase quantization. 
\subsection{Ideal (continuous) phase shifters}
\label{ideal phase shifters}
We now briefly review some asymptotic results that will be particularly useful in our analysis.  
\begin{Lindeberg}
\label{Lindeberg}
Let $\mathbf{p}$ and $\mathbf{q}$ be two $n \times1$ mutually independent vectors whose elements are i.i.d RVs with variances $\sigma_{p}^2$ and $\sigma_{q}^2$, respectively. Then, based on the law of large numbers, we have
\begin{align}
\frac{1}{n}{\mathbf{p}^H\mathbf{p}\stackrel{a.s.}{\longrightarrow}\sigma_{p}^2}\,\,\,\, ,{\rm and} \,\,\,\,\,       \frac{1}{n}\mathbf{q}^H\mathbf{q}\stackrel{a.s.}{\longrightarrow}\sigma_{q}^2 \,\,\,\, ,{\rm as}\,\,\,\,\, n\rightarrow \infty
\end{align}
where $\stackrel{a.s.}{\longrightarrow}$ indicates almost sure convergence. Moreover, based on the Lindeberg--L\'evy central limit theorem we can write
\begin{align}
\frac{1}{\sqrt{n}} {\mathbf{p}^H\mathbf{q}\stackrel{\mathrm{dist.}}{\longrightarrow}\mathcal{CN}(0,\sigma_{p}^2\sigma_{q}^2)}
\end{align}
where $\stackrel{\mathrm{dist.}}{\longrightarrow}$ shows the convergence in distribution.
\end{Lindeberg}
We can now turn our attention to the analog processing matrices $\mathbf{F}_1$ and $\mathbf{F}_2$ which satisfy the following relationship
\begin{F1}
\label{F1}
As $N \to \infty$, the matrices $\mathbf{F}_1\mathbf{F}_1^H$ and also $\mathbf{F}_2\mathbf{F}_2^H$ converge pairwise to the identity matrix as follows\begin{align}
&\mathbf{F}_1\mathbf{F}_1^H\stackrel{a.s.}{\longrightarrow}\mathbf{I}_{K_r} \nonumber \\
&\mathbf{F}_2\mathbf{F}_2^H\stackrel{a.s.}{\longrightarrow}\mathbf{I}_{K_t}.
\end{align}
\end{F1}
\begin{proof}
See Appendix \ref{Lemma F1}.
\end{proof}
\begin{Fn}
\label{Fn}
As $N \to \infty$, the analog phase shifter, $\mathbf{F}_1$, preserves the distribution of the AWGN noise due to its orthonormal rows. 
\end{Fn}
\begin{FH}
\label{FH}
Let us define $\mathbf{I}_{a,b,r} \in \mathbb{C}^{a\times b}$ as an $a\times b$ diagonal matrix whose first $r$ elements on the main diagonal are $1$, and the rest are $0$. Then, 
\begin{align}
\mathbf{F}_1\mathbf{H}_1\stackrel{a.s.}{\longrightarrow}\sqrt{\frac{N\pi}{4}}\mathbf{I}_{K_r,K,r_1} \nonumber \\
\mathbf{F}_2\mathbf{H}_2\stackrel{a.s.}{\longrightarrow}\sqrt{\frac{N\pi}{4}}\mathbf{I}_{K_t,K,r_2} 
\end{align}
where $r_1=\min{\left(K_r,K\right)}$ and $r_2=\min{\left(K_t,K\right)}$.
\end{FH}
\begin{proof}
See Appendix \ref{LemmaFH}.
\end{proof}
Turning now to (\ref{Matrix Model}) and using the aforementioned lemmas, when $N \to \infty$, it can be shown that
\begin{align}
\mathbf{y}_D \rightarrow &\sqrt{P_u}\alpha\Big(\frac{N\pi}{4}\Big)^2\big(\mathbf{D}_2^{\frac{1}{2}}\big)^H\mathbf{D}_2^{\frac{1}{2}}\big(\mathbf{D}_1^{\frac{1}{2}}\big)^H\mathbf{D}_1^{\frac{1}{2}} \mathbf{x} \nonumber 
\\+& \alpha\Big(\frac{N\pi}{4}\Big)^{\frac{3}{2}}\big(\mathbf{D}_2^{\frac{1}{2}}\big)^H\mathbf{D}_2^{\frac{1}{2}}\big(\mathbf{D}_1^{\frac{1}{2}}\big)^H\mathbf{n}_R+\mathbf{n}_D
\end{align}
which can be simplified for the $k$-th destination as
\begin{align}
y_{D_k}\hspace{-2pt}\rightarrow \hspace{-2pt}\sqrt{P_u}\alpha\Big(\frac{N\pi}{4}\Big)^2 \eta_{2_k}\eta_{1_k} x_k \hspace{-2pt}+\hspace{-2pt} 
\alpha\Big(\frac{N\pi}{4}\Big)^{\frac{3}{2}}\eta_{2_k}\eta_{1_k}^{\frac{1}{2}}{n}_{R_k}\hspace{-2pt}+\hspace{-2pt}n_{D_k}
\end{align}
where $r=\min\left(K_r,K_t,K\right)$ and $k\in \left\{1,2,...,r\right\}$. Thus, from (\ref{Monte Carlo}) we can obtain the corresponding SINR for the $k$-th destination in the case that the number of antennas increases without bound
\begin{align}
\label{SINR}
{\rm SINR}_k \rightarrow \frac{(\frac{N\pi}{4})^4 P_u\alpha^2\eta_{1_k}^2\eta_{2_k}^{2}}{(\frac{N\pi}{4})^3\sigma_{n_R}^2\alpha^2\eta_{1_k}\eta_{2_k}^{2}+\sigma_{n_D}^2}.
\end{align}

In the following, we investigate three power scaling strategies and draw very interesting engineering insights. Our analysis can be divided into three main cases, namely, Case \ref{caseI}) fixed $NP_u$ and $NP_r$ while $N\to\infty$; Case \ref{caseII}) fixed $NP_u$ while $N\to\infty$; Case \ref{caseIII}) fixed $NP_r$ while $N\to\infty$.

\begin{enumerate}
	\item \label{caseI} Let $\lim\limits_{N\to\infty} NP_u=E_u$ and $\lim\limits_{N\to\infty} NP_r=E_r$ where both $E_u$ and $E_r$ are  finite constants. Then, from (\ref{alpha}) we can get
\begin{align}
N^3\alpha^2\rightarrow\frac{E_r}{(\frac{\pi}{4})^3 E_t\sum\limits_{i=1}^{r}{\eta_{1_i}^2\eta_{2_i}}+(\frac{\pi}{4})^2 \sigma_{n_R}^2\sum\limits_{i=1}^{r}{\eta_{1_i}\eta_{2_i}}}
\end{align}
which finally yields (\ref{big SINR}) shown at the top of next page.
\begin{figure*}[t!]
\begin{align}
\label{big SINR}
{\rm SINR}_k\rightarrow\frac{\big(\frac{\pi}{4}\big)^2 E_u E_r\eta_{1_k}^2\eta_{2_k}^2}{\big(\frac{\pi}{4}\big)E_r\sigma_{n_R}^2\eta_{1_k}\eta_{2_k}^2+\big(\frac{\pi}{4}\big)E_u\sigma_{n_D}^2\sum\limits_{i=1}^{r}{\eta_{1_i}^2\eta_{2_i}}+\sigma_{n_R}^2\sigma_{n_D}^2\sum\limits_{i=1}^{r}{\eta_{1_i}\eta_{2_i}}} \;\;\;\; k\in\left\{1,2,\cdots,r\right\}
\end{align}
\hrule
\end{figure*}
As a consequence, under a full CSI assumption we can reduce the transmitted power and also relay power proportionally to $\frac{1}{N}$ if the number of relay antennas grows without bound. This result is consistent with \cite{ngo2013energy}.

\item \label{caseII} Let $\lim\limits_{N\to\infty} NP_u=E_u$, where $E_u$ is a finite constant. Then, returning to (\ref{big SINR}) and after a few simplifications we obtain 
\begin{align}
{\rm SINR}_k\rightarrow\frac{\pi}{4}\frac{E_u\eta_{1_k}}{\sigma_{n_R}^2}
\end{align} 
which is associated with the following average spectral efficiency
  
\begin{align}
R_2\rightarrow \frac{1}{2}\sum\limits_{k=1}^{r}\log_2\Big(1+\frac{\pi}{4}\frac{E_u\eta_{1_k}}{\sigma_{n_R}^2}\Big).
\end{align}

The above result is quite intuitive. It shows that if the number of RF chains is, at least, equal to the number of users, i.e. $\min(K_r,K_t)\geq K$ or equivalently $r=K$,  we can enjoy full multiplexing gain and boost the achievable spectral efficiency. Moreover, in comparison with a single-input single-output (SISO) system without any intra-cell interference, our system model only suffers a $\frac{\pi}{4}$-fold reduction on the power gain due to the analog processing. All in all, this power gain penalty is quite acceptable as we have eliminated many relay RF chains, and consequently, we have substantially reduced the circuitry complexity and power consumption. 
Similar to Case 1, we can infer that we can scale down the transmit power analogously to the number of relay antennas and, still, maintain a non-zero spectral efficiency.   
\item \label{caseIII} Let $\lim\limits_{N\to\infty} NP_r=E_r$, where $E_r$ is a finite constant. Then, we can find out the average spectral efficiency in the same way as pointed out in Case \ref{caseII} to get

\begin{align}
R_3\rightarrow \frac{1}{2}\sum\limits_{k=1}^{r}\log_2\Bigg(1+\frac{\pi}{4}\frac{E_r\eta_{1_k}^2\eta_{2_k}^2}{\sigma_{n_D}^2\sum\limits_{i=1}^{r}{\eta_{1_i}^2\eta_{2_i}}}\Bigg).
\end{align}
It is noteworthy that if we ignore large-scale fading effects, we get the same results in Case \ref{caseII} and \ref {caseIII}. However, Case \ref{caseIII} converges faster than Case \ref{caseII} to its own asymptotic result. This can be observed from (\ref{big SINR}). In Case \ref{caseII}, we can ignore the constant term $E_r\sigma_{n_R}^2\eta_{1_k}\eta_{2_k}^2$ in comparison with $E_u\sigma_{n_D}^2\sum_{i=1}^{r}{\eta_{1_i}^2\eta_{2_i}}$ even for moderate number of antennas. In contrast, in Case \ref{caseII}, a much higher number of antennas is required to ignore the constant term $E_u\sigma_{n_D}^2\sum_{i=1}^{r}{\eta_{1_i}^2\eta_{2_i}}$ vs. the scaled term $E_r\sigma_{n_R}^2\eta_{1_k}\eta_{2_k}^2$ in (\ref{big SINR}).     
\end{enumerate}
\subsection{Phase Quantization}
\label{quantized phase shifters}
Until now, we have assumed ideal analog phase shifters (beamformers) which generate any required phases. However, the implementation of such shifters with continuous phase is not feasible or, at least, is quite expensive due to hardware limitations \cite{ni2015near,el2014spatially,liang2014low,ying2015hybrid}. 
Most importantly, quantized analog beamformers are more attractive in limited feedback systems \cite{alkhateeb2013hybrid,roh2006design}. In the rest of this paper, the system performance will be assessed under quantized phases. Thus, the phase of each entry of $\mathbf{F}_1$ and $\mathbf{F}_2$ is chosen from a codebook $\Psi$ based on the closest Euclidean distance.
\begin{align}
\Psi=\Big\{0,\pm\Big(\frac{2\pi}{2^\beta}\Big),\pm2\Big(\frac{2\pi}{2^\beta}\Big),\cdots,\pm 2^{\beta-1}\Big(\frac{2\pi}{2^\beta}\Big)\Big\}
\end{align}   
where, $\beta$ denotes the number of quantization bits. As pointed out previously, the channel coefficients $\big[\mathbf{G}_1\big]_{m,n}$ and $\big[\mathbf{G}_2\big]_{m,n}$ all have uniform phase between $0$ and $2\pi$, such that $\angle\big[\mathbf{G}_i\big]_{m,n}=\phi_{m,n} \sim U(0,2\pi)$, for $i=1,2$. 
Let us define $\epsilon_{m,n}$ as the error between the unquantized phase $\phi_{m,n}$ and quantized phase $\hat{\phi}_{m,n}$ chosen from the codebook
\begin{align}
\epsilon_{m,n}\stackrel{\Delta}{=}\phi_{m,n}-\hat{\phi}_{m,n}.
\end{align}
Due to the uniform distribution of phase, we can easily conclude that the error is an uniform RV, i.e. $\epsilon_{m,n}\sim U[-\delta ,+\delta\big)$, where we define $\delta\stackrel{\Delta}{=}\frac{\pi}{2^\beta}$. This error affects Lemma \ref{FH}, and in turn, the average spectral efficiency. For this reason, we provide the following lemma to account for phase quantization.\footnote{Hereafter, we use a hat sign for the variables that are associated with the quantized beamforming assumption.} 

\begin{FHextension}
\label{FHextension}
Let $\hat{\mathbf{F}}_1$ and $\hat{\mathbf{F}}_2$ denote the analog detector and precoder, respectively. Then,
\begin{align}
&\hat{\mathbf{F}}_1\mathbf{H}_1\stackrel{a.s.}{\longrightarrow}\sqrt{\frac{N\pi}{4}}\sinc(\delta)\mathbf{I}_{K_r,K,r_1} \nonumber \\
&\hat{\mathbf{F}}_2\mathbf{H}_2\stackrel{a.s.}{\longrightarrow}\sqrt{\frac{N\pi}{4}}\sinc(\delta)\mathbf{I}_{K_t,K,r_2} 
\end{align}
where we define $\sinc(\delta)\stackrel{\Delta}{=}\frac{\sin(\delta)}{\delta}$.
\end{FHextension}

\begin{proof}
The results follow trivially by using the \mbox{methodology} outlined in Appendix \ref{LemmaFH}.
\end{proof}

Now, we incorporate Lemma \ref{FHextension} into the system model and signal description. The modified normalization factor $\hat{\alpha}$ can be found at (\ref{big formula 2}) on the top of next page. Furthermore, the received signal for the $k$-th destination can be obtained from the following formula 
\begin{figure*}[t!]
\begin{align}
\label{big formula 2}
\hat{\alpha} \rightarrow \sqrt{\frac{P_r}{P_u\sinc^6(\delta)\|\mathbf{F}_2^H\mathbf{A}_2 \mathbf{A}_1^H\mathbf{A}_1\|_{F}^2 + \sigma_{n_R}^2 \sinc^4{(\delta)}\|{\mathbf{F}_2^H\mathbf{A}_2\mathbf{A}_1^H\mathbf{F}_1\|}_F^2}}.
\end{align}
\end{figure*}

\begin{align}
  \hat{y}_{D_k}\rightarrow &\sqrt{P_u}\sinc^4(\delta)\hat{\alpha}\Big(\frac{N\pi}{4}\Big)^2 \eta_{2_k}\eta_{1_k} x_k \\ \nonumber
+&\Big(\frac{N\pi}{4}\Big)^{\frac{3}{2}}\sinc^3(\delta)\hat{\alpha}\,\eta_{2_k}\eta_{1_k}^{\frac{1}{2}}{n}_{R_k}+n_{D_k}.
\end{align}  
Phase quantization also affects the power scaling strategies considered in Cases \ref{caseI}--\ref{caseIII} above. The corresponding results for these three cases under quantized analog processing can be modified as shown in (\ref{big formula 3}) (on the top of next page), (\ref{case 2}) and (\ref{case 3}), respectively.
\begin{figure*}[t!] 
\begin{align}
\label{big formula 3}
\hat{R}_1 \rightarrow \frac{1}{2} \log_2 \Bigg(1+\frac{(\frac{\pi}{4})^2\sinc^8(\delta) E_u E_r\eta_{1_k}^2\eta_{2_k}^2}{(\frac{\pi}{4})\sinc^6(\delta)E_r\sigma_{n_R}^2\eta_{1_k}\eta_{2_k}^2
+(\frac{\pi}{4})\sinc^6(\delta)E_u\sigma_{n_D}^2\sum\limits_{i=1}^{r}{\eta_{1_i}^2\eta_{2_i}}
+\sinc^4(\delta)\sigma_{n_R}^2\sigma_{n_D}^2\sum\limits_{i=1}^{r}{\eta_{1_i}\eta_{2_i}}}\Bigg).
\end{align}
\hrule
\end{figure*}
\begin{align}
\label{case 2}
\hat{R}_2 \rightarrow \frac{1}{2}\sum\limits_{k=1}^{r}\log_2\Big(1+\frac{\pi}{4}\frac{E_u\eta_{1_k}}{\sigma_{n_R}^2}\sinc^2(\delta)\Big).
\end{align}
\begin{align}
\label{case 3}
\hat{R}_3 \rightarrow \frac{1}{2}\sum\limits_{k=1}^{r}\log_2\Bigg(1+\frac{\pi}{4}\frac{E_r\eta_{1_k}^2\eta_{2_k}^2}{\sigma_{n_D}^2\sum\limits_{i=1}^{r}{\eta_{1_i}^2\eta_{2_i}}}\sinc^2(\delta)\Bigg).
\end{align}  
Taken together, these results indicate a penalty function associated with quantized processing. Roughly speaking, $\sinc^2(\delta)$ is a good approximation of this power gain penalty. In a worst case, where we have only one quantization bit $\beta=1$, the SINR will be reduced by a factor of $\sinc^2(\frac{\pi}{2})=\frac{4}{\pi^2}\approx 40\%$. 
As pointed out in \cite{el2014spatially}, a reasonable rule-of-thumb is to add $1$ bit resolution while the number of antennas doubles, since beam width is inversely relative to the number of antennas.            
\section{Simulation Results}
In this section, Monte-Carlo simulations are provided to assess the validity of the average spectral efficiency of a multipair relay system. We assume that the relay covers a circular area with a radius of $1000$ meters. Users are located with a uniform random distribution around the relay with a guard zone of $r_g=100$ meters. We consider a Rayleigh flat fading channel for small-scale fading effects.
\begin{figure}[t!]
	\centering
		\includegraphics[width=3.48in]{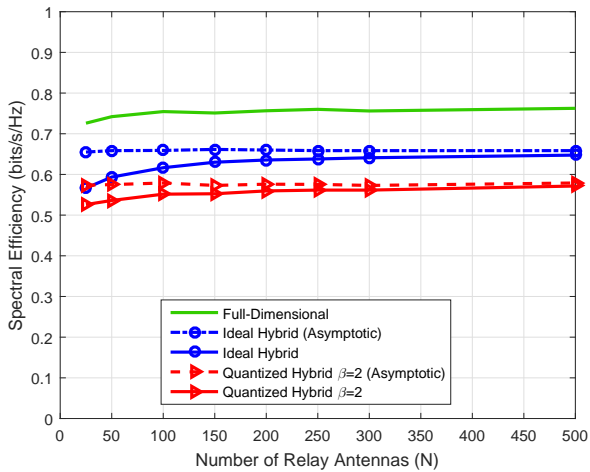}
		\caption{Average spectral efficiency in Case \ref{caseI} $\big(E_u=E_r=13 \; \mathrm{dB}\big)$.}
\label{case_1}
\end{figure} 
\begin{figure}[t!]
	\centering
		\includegraphics[width=3.48in]{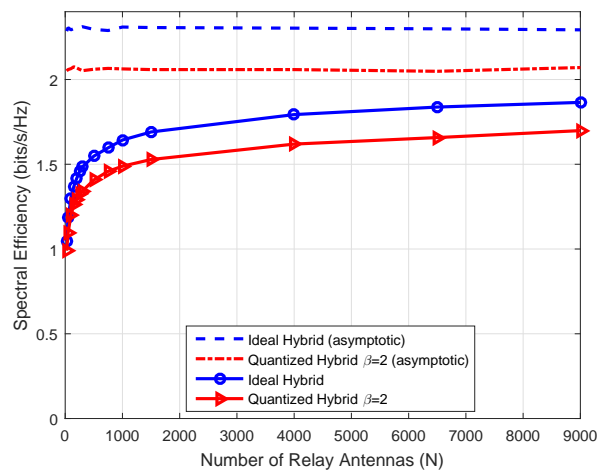}
		\caption{Average spectral efficiency in Case \ref{caseII} $\big(E_u=13
		\;\mathrm{dB} , P_r=13 \; \mathrm{dB}\big)$.}
\label{case_2}
\end{figure}
Also, the large-scale fading is modeled via a log-normal RV, with standard deviation $\sigma_{sh}$, which is multiplied by $\big(\frac{r_k}{r_g}\big)^{\nu}$ to model the path-loss as well. Here, $r_k$ is the distance between the $k$-th user and the relay, and also $\nu$ denotes the path loss exponent. Without loss of generality, we set $\sigma_{n_R}^2=\sigma_{n_D}^2=1$, $\nu=3.8$, $K_r=K_t=K=10$ and $\sigma_{sh}=8$ dB for all simulations.

Figure \ref{case_1} compares the performance of full-dimensional topology, where all amount of detection/precoding is performed in the digital domain, against that of hybrid topology with continuous and quantized analog processing. A full-dimensional massive relay is equipped by $N$ RF chains which seems to be infeasible in practice, while this number is reduced to only $K=10$ in the hybrid structure. Moreover, a hybrid relay deploys two inexpensive beamformers which can be actually implemented in the analog domain with phase shifters. It can be also observed that the hybrid scheme performs very close to the conventional scheme, with about a $10 \%$ reduction in spectral efficiency but substantially reduced complexity. However, this reduction in spectral efficiency can be compensated by deploying more antennas at the relay without any additional RF chains. Hence, this promising idea seems to be a viable alternative to conventional relaying topologies. Moreover, this figure examines a more restricted case, where there is a severe phase control on beamformers with only $2$ bit resolution. Results confirm that the proposed method suffers a negligible reduction.
Figures \ref{case_2} and \ref{case_3} demonstrate similar results for Case $2$ and $3$, respectively. Clearly, as the number of relay antennas increases, the average spectral efficiency approaches to the saturation value which is expected by our power scaling laws. Note also that the curve scales much slower in Case $2$ in comparison with Case \ref{caseI} and \ref{caseIII}.
\section{Conclusion}
Massive MIMO is a major candidate for the next generation of wireless systems. This technique combined with relays can enhance the cell coverage while it enjoys a simple signal processing at the relay. On the other hand, the high cost and power consumption of RF chains can be prohibitive due to the large number of mixers and power amplifiers. For this reason, we used an analog/digital (hybrid) structure at the relay and also reduced the number of RF chains analogously to the number of users while the system still enjoys a full multiplexing gain. Finally, we analytically quantified the system spectral efficiency and demonstrated a great performance of the proposed configuration even under coarse quantization.
\appendices
\section{Proof of Lemma \ref{F1}}
\label{Lemma F1}
Having discussed how to construct  $\mathbf{F}_1 \in \mathbb{C}^{K_r \times N}$, we can write each entry of this matrix as $\frac{1}{\sqrt{N}}\exp{\left(j\theta_{m,n}\right)}$, where $\theta_{m,n}$ is a uniform RV, i.e. $\theta_{m,n}\sim U\left[0,2\pi\right)$. Now, let the vectors $\mathbf{f}_{1_p}$ and $\mathbf{f}_{1_q}$ denote the $p$-th and $q$-th rows of matrix $\mathbf{F}_1$, respectively. Then,
$\mathbf{f}_{1_p}\mathbf{f}^H_{1_p}=1$ since the phases cancel out each other. On the other hand, if $N \to \infty$, due to the central limit theorem for any $p\neq q$ we have
\begin {align}
   \mathbf{f}_{1_p}\mathbf{f}^H_{1_q}=\frac{1}{N}\sum\limits_{l=1}^{N}{\mathrm{e}^{j\left(\theta_{p,l}-\theta_{q,l}\right)}} \rightarrow \mathbb{E}\Big[{\rm e}^{j\theta_p}\Big]\mathbb{E}\Big[\mathrm{e}^{-j\theta_q}\Big]=0
\end{align}
where the distribution of $\theta_p$ and $\theta_q$ are defined similar to $\theta_{m,n}$. Likewise, we can prove the second part.
\section{Proof of Lemma \ref{FH}}
\label{LemmaFH}
Let us rewrite the $\left(m,n\right)$-th entry of matrix ${\mathbf{H}}_1 \in \mathbb{C}^{N \times K}$ like $r_{m,n}\mathrm{e}^{j\phi_{m,n}}$, where the amplitude and phase have a Rayleigh and uniform distribution, respectively. In other words, $r_{m,n} \sim R\left(\sqrt{\frac{1}{2}}\right)$, and $\phi_{m,n} \sim U\left[0,2\pi\right)$. Now, let the vectors $\mathbf{f}_{1_p}$ and $\mathbf{h}_{1_p}$ denote the $p$-th row of matrix $\mathbf{F}_1$ and $p$-th column of matrix $\mathbf{H}_1$, respectively. Then, since phases cancel out each other, for any $p\leq r_1$ we have that
\begin{align}
\mathbf{f}_{1_p}\mathbf{h}_{1_p}=\frac{1}{\sqrt{N}}\sum\limits_{l=1}^{N}r_{p,l}\stackrel{(a)}{\rightarrow}\sqrt{N}\mathbb{E}\big[r_p\big]=\sqrt{\frac{N\pi}{4}}
\end{align} 
where we have used the central limit theorem in $\left(a\right)$, and the fact that $r_p$ is a Rayleigh RV with parameter $\sqrt\frac{1}{2}$. We can also prove the second part in a similar way.
\begin{figure}[t!]
	\centering
		\includegraphics[width=3.48 in]{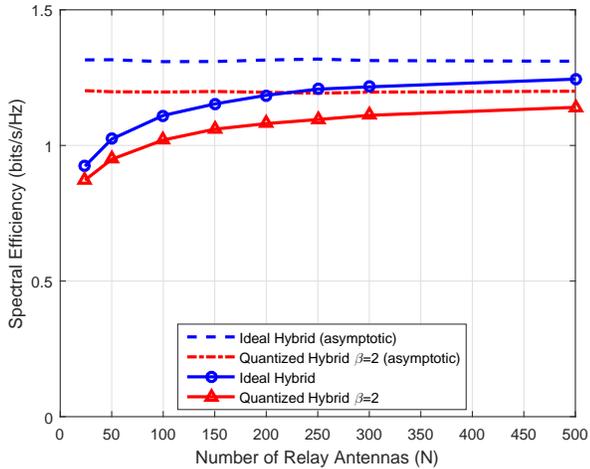}
		\caption{Average spectral efficiency in Case \ref{caseIII} $\big(P_u=13\;\mathrm{dB}, E_r=13 \; \mathrm{dB}\big)$.}
\label{case_3}
\end{figure}
\section*{Acknowledgment}
This work of S. Jin was supported by the \mbox{National} Natural Science Foundation of China under Grants $\left(61531011, 61450110445\right)$.
\bibliographystyle{IEEEtran}
\bibliography{IEEEabrv,refs}
\end{document}